\documentclass[review,onefignum,onetabnum]{siamart190516}
\usepackage{algorithm}
\usepackage{algorithmic}
\usepackage{graphicx}
\usepackage[margin=1in]{geometry}
\usepackage{lipsum}
\usepackage{amsfonts}
\usepackage{graphicx}
\usepackage{epstopdf}
\usepackage{algorithmic}
\ifpdf
  \DeclareGraphicsExtensions{.eps,.pdf,.png,.jpg}
\else
  \DeclareGraphicsExtensions{.eps}
\fi

\newsiamremark{remark}{Remark}
\newsiamremark{hypothesis}{Hypothesis}
\crefname{hypothesis}{Hypothesis}{Hypotheses}
\newsiamthm{claim}{Claim}

\headers{Parallel coarsening with preserved spectra}{C. Brissette, A. Huang, and G. Slota}

\title{Parallel coarsening of graph data with spectral guarantees\thanks{\funding{This work is funded under Laboratory Directed Research and Development (LDRD) program, Project No. 218474, at Sandia National Laboratories. Sandia National Laboratories is a multimission laboratory managed and operated by National Technology and Engineering Solutions of Sandia, LLC., a wholly owned subsidiary of Honeywell International, Inc. for the U.S. Department of Energy’s National Nuclear Security Administration under contract DE-NA-0003525. This work was also supported by the U.S. Department of Energy, Office of Science, under award DE-AC52-07NA27344 (FASTMath SciDAC Institute). This paper describes objective technical results and analysis. Any subjective views or opinions that might be expressed in the paper do not necessarily represent the views of the U.S. Department of Energy or the United States Government.}}}

\author{Christopher Brissette\thanks{Rensselaer Polytechnic Intitute, Troy, NY 
  (\email{brissc@rpi.edu}).}
\and Andy Huang\thanks{Sandia National Labs, Albuquerque, NM  
  (\email{ahuang@sandia.gov}).}
\and George Slota\thanks{Rensselaer Polytechnic Intitute, Troy, NY (\email{slotag@rpi.edu}).}}

\usepackage{amsopn}

\begin{document}
\nolinenumbers
\maketitle
\begin{abstract}
Finding coarse representations of large graphs is an important computational problem in the fields of scientific computing, large scale graph partitioning, and the reduction of geometric meshes. Of particular interest in all of these fields is the preservation of spectral properties with regards to the original graph. While many methods exist to perform this task, they typically require expensive linear algebraic operations and yield high work complexities. We adapt a spectral coarsening bound from the literature in order to develop a coarsening algorithm with a work complexity that is drastically smaller than previous work. We further show that this algorithm is easily parallelizable and presents impressive scaling results on meshes.
\end{abstract}



\section{Introduction}
The coarsening of graphs while preserving properties of the original graph has been widely studied in many contexts~\cite{chen2021graph}. Starting with Gabriel Kron's seminal work coarsening resistor networks~\cite{kron1939tensor}, the study has flourished to have uses in computer graphics, multilevel partitioning, and broader network science topics. While there are many different metrics one may preserve when performing coarsening, in this work we focus on coarsening graphs so as to preserve both the smallest eigenvalues and eigenvectors of the Laplacian of the original graph. These techniques are particularly compelling given that graph structures can be recovered from arbitrary data through constructions such as the Vietoris Rips complex~\cite{vietoris1927hoheren}.\\
\indent Several methods exist for performing the task of spectrum preserving coarsening, however they are often computationally expensive. For instance in Liu et al.~\cite{liu2019spectral} a method is suggested where an initial combinatorial coarsening is performed and then cluster assignment is refined by iterating upon and optimizing a given matrix norm. This method does not have an explicit bound on its run time and instead iterates gradient descent to within some tolerance. Alternatively, Loukas~\cite{loukas2019graph} suggests a ``local variation" method where $k$ eigenvalues are pre-computed and then used as a point of comparison when considering potential merges. The explicit complexity of this method is given by $\Tilde{O}(ckm + k^2n + ck^3 +\sum_{l=1}^{c}\Phi_{l}(min\{k^2\delta + k\delta^2, k\delta^2+\delta^3\}+log|\mathcal{F}_l|))$. Here $c$ is the number of coarsening levels, $k$ is the number of desired eigenvectors matched, $n$ is the number of nodes, and $m$ is the number of edges in the original graph. For further details bout this complexity, one can visit the original publication. Both of these methods boast powerful spectral guarantees and impressive results. However, as discussed, they inherently have limitations in terms of compute time due to both the usage of expensive linear algebra operations as well as the methods being inherently sequential and unable to make use of modern parallel architectures. Other spectrum preserving coarsening methods~\cite{lescoat2020spectral,chen2020chordal} fall victim to similar issues.\\
\indent We will first define the notion of a coarsened graph. Note that we will assume all graphs $G = (V,E)$ to be connected with no nodes of degree zero, the number of nodes is $|V| = n$, the number of edges is $|E| = m$, and $n_c$ is the desired number of nodes in our coarsening.
\begin{definition}
Given a graph $G=(V,E)$, a \textbf{coarsened graph} $G_c=(V_c,E_c)$ is formed by merging nodes in $\{v_1,\cdots,v_k\} \subseteq V$ into super-nodes $v_c \in V_c$. The associated edge weights for $v_c$ are given as $a_{v_c x} = \sum_{v_i \in v_c} w_{v_i x}$.
\end{definition}
Note that the Laplacian of $G$ has a higher dimension than the Laplacian of its coarsened counterpart $G_c$. This makes comparing spectral properties difficult. For this purpose we introduce the lift of a coarsened graph. Intuitively this object takes a coarsened graph $G_c$ and projects it to the same dimension as $G$. In doing so it provides us with a convenient way to compare eigenvalues and eigenvectors of the original graph as well as the coarsened graph.
\begin{definition}
Given a coarsened graph $G_c = (V_c,E_c)$ of $G = (V,E)$, its \textbf{lift} $\hat{G} = (\hat{V},\hat{E})$ is defined to be a graph on $n$ nodes, where every super node in $V_c$ is expanded to its original number of nodes in $V$, and the weight between each pair of nodes $u,v\in \hat{V}$ is $\hat{w}(u,v)= \frac{ \sum_{x \in u_c } \sum_{y \in v_c} w(x,y) }{|u_c||v_c|} $. Here $u_c,v_c \in V_c$ are the super nodes that $u,v\in V$ belong to in the coarsened graph, and $|\cdot|$ denotes the number of nodes within those super nodes. Additionally $w(a,b)$ is the weight of the connection between nodes $a,b \in V$.
\end{definition}
We interest ourselves with a bound presented by Jin and Loukas~\cite{jin2020graph}. They derive the following relationship between our original graph $G$, and its lift $\hat{G}$ after coarsening.
\begin{theorem}
    For a graph $G=(V,E)$, if two nodes $u,v \in V$ are such that $\|\frac{w_u}{d_u}-\frac{w_v}{d_v}\|_1 \leq \epsilon$ and we merge them, $\|\Lambda - \hat{\Lambda}\|_\infty \leq \epsilon$.
    \label{thm:approx1}
\end{theorem}
Here, $\Lambda$ is a sorted vector of eigenvalues of the normalized Laplacian $\mathcal{L}(G)$~\cite{chung1997spectral}, $w_x$ is the weighted adjacency vector of node $x \in V$, and $d_x$ is the degree of node $x \in V$~\cite{godsil2001algebraic}. $\hat{\Lambda}$ is the sorted vector of eigenvalues of the normalized Laplacian of the lift $\mathcal{\hat{G}}$ after merging the nodes $u,v\in V$ in the original graph. This bound is important since every eigenvalue of our coarsened graph is also an eigenvalue of the lift~\cite{jin2020graph}.\\

\section{A new algorithm}
\indent In~\cite{jin2020graph} a greedy algorithm for spectrum consistent coarsening is defined. This algorithm, which has a complexity of $O(m(n+n_c)(n-n_c))$ will be called the ``explicit greedy algorithm" for the remainder of this paper. The algorithm works as follows. For every edge in the graph the 1-norm is computed between the adjacency vectors of the nodes at the ends of that edge. Then these 1-norms are sorted and a merge is performed along the edge with the smallest 1-norm. These steps are then repeated to our desired level of coarsening. First we note how the inner two loops of the explicit greedy algorithm may be parallelized. These loops iterate over each of $m$ edges and compute the associated norm-difference between node adjacencies. Therefore we are computing $O(m)$ norms, where the computation of each norm between nodes $u,v\in V$ has sparse vector work complexity $O(d_u + d_v)$. This gives us an overall complexity for these inner loops of $O(n\left<d^2\right>)$ where $\left<d^2\right>$ is the second moment of the degree distribution of our graph.  Note that none of these norms depend on previously computed norms, so these inner loops can be parallelized such that given $p$ available threads, if we evenly distribute the edges among threads we achieve a shared-memory parallel-time complexity of $O(\frac{n}{p}\left<d^2\right>)$ for our inner loops, where $p$ is the number of threads. \\
\indent To achieve further complexity reduction we would like to be able to remove the outer loop of the greedy algorithm. This will allow us to avoid recomputing norm differences. For this purpose we make an additional observation that the eigenvalue differences of an arbitrary level of coarsening may be bounded in terms of the norm differences in the original uncoarsened graph given by the following theorem. 
\begin{theorem}
Given a set of $s$ merges $\{(a_1,b_1),\cdots,(a_s,b_s)\}$ where $\|\frac{w_{a_i}}{d_{a_i}} - \frac{w_{b_i}}{d_{b_i}}\|_1 \leq \epsilon$ for every $i\in[1..s]$, then we know  $\|\Lambda - \hat{\Lambda}\|_\infty \leq \frac{s(s+1)}{2}\epsilon$.
\label{thm:approx2}
\end{theorem}
\begin{proof}
The proof breaks in to two parts. First we wish to show that $\|\frac{w_a}{d_a}-\frac{w_b}{d_b}\|_1 \leq \epsilon$ implies, without loss of generality, that $\|\frac{w_a}{d_a}-\frac{w_a+w_b}{d_a+d_b}\|_1 \leq epsilon$. This statement says that if two nodes have a small adjacency difference norm before merging, each of them will also have a small difference norm when compared with the merged node. Take $v=\frac{w_a}{d_a} - \frac{w_a+w_b}{d_a+d_b}$, then by algebra it can be shown that $\|v\|_1 \leq \frac{\epsilon}{\frac{d_a}{d_b}+1} \leq \epsilon$, which directly implies $\|\frac{w_a}{d_a}-\frac{w_a+w_b}{d_a+d_b}\|_1 \leq \epsilon$ \\

For the second half of our proof, imagine we have performed some number of merges. Then there are two possibilities for the subsequent merge. Either it shares no nodes with the previous merge, or it shares nodes with the previous merge. In the first case, neither of the nodes have been coarsened, so their similarity remains the same and theorem~\ref{thm:approx1} tells us that we simply add $\epsilon$ to our spectral error bound after merging them. In the latter case, assume without loss of generality that our merge consists of nodes $a,c$ where $a$ overlaps with the first merge. Then by triangle inequality, and the first part of our proof we know $\|\frac{w_c}{d_c} - \frac{w_a + w_b}{d_a + d_b}\|_1 \leq 2\epsilon$. This means that the normalized adjacency vector of node $c$ is $2\epsilon$ similar to the normalized adjacency vector of the supernode given by our first merge. Therefore we can apply Theorem~\ref{thm:approx1} again and this error of $2\epsilon$, adds to our bound. We may then repeat this process until our desired level of coarsening, adding a factor of $\epsilon$ to each addition. The worst case error occurs when all the merges overlap, in which case the upper bound is given by the following.
\begin{align*}
    \|\Lambda - \hat{\Lambda}\|_\infty &\leq \sum_{k=1}^{s}k\epsilon = \frac{s(s+1)}{2}\epsilon
\end{align*}

\end{proof}

While looser than the bound provided in~\cite{jin2020graph}, this bound requires no knowledge about intermediate coarsenings of our graph. Because of this, we can remove the outer loop of the explicit greedy algorithm to obtain an order $O(n)$ work complexity reduction while still preserving a spectral bound. This gives us our Algorithm~\ref{alg:approxgreedy}.

\begin{algorithm}
\caption{Approximate Greedy Coarsen (G = (V,E), $n_c$, p)}\label{alg:approxgreedy}
\begin{algorithmic}
\STATE $s \gets |V|$
\STATE $\mathit{fitness} \gets \emptyset$
\STATE $\{E_1,\cdots,E_p\} \gets \mathit{partition}( E, p )$

\FOR{$(a,b) \in E_i$ \textbf{in parallel}} 
    \STATE $\mathit{fitness} \gets \mathit{fitness} \cup ( (a,b), \|\frac{w_a}{d_a} - \frac{w_b}{d_b}\|_1 )$
\ENDFOR

\STATE $\mathit{fitness} \gets \mathit{sort}( \mathit{fitness} )$ 

\WHILE{ $s < n_c$ }
    \STATE $(a,b) \gets \mathit{fitness}[ \mathit{index} ][ 0 ] $
    \STATE $G \gets \mathit{merge}( a, b )$
    \STATE $s \gets s - 1$
\ENDWHILE
\end{algorithmic}
\end{algorithm}

Algorithm~\ref{alg:approxgreedy} closely resembles the explicit greedy algorithm with some reorganization. We still iterate through every edge and calculate a fitness function which is the norm-difference of the normalized adjacencies between each node, however we only perform this once for each edge. Then a sort is performed, which can be done in $O(\frac{m}{p}log(m))$ parallel time. Finally, our merges are performed in order of fitness, requiring $O(n - n_c)$ operations. This gives us a final parallel time of $O( \frac{n}{p}\left<d^2\right> + \frac{m}{p}log(m) + (n-n_c) )$ for Algorithm~\ref{alg:approxgreedy}.

\section{Results}
We present results for both the scaling and approximation properties of our algorithm on two small example graphs, the ``Stanford bunny" from the Stanford 3D Scanning Repository\footnote{\url{http://graphics.stanford.edu/data/3Dscanrep/}} as well as ``ego-Facebook" from the Stanford Large Network Dataset Collection\footnote{\url{https://snap.stanford.edu/data/}}. The former is a nearly-regular mesh and the latter has a highly irregular degree distribution. We find that while spectral approximations appear better on the irregular ``ego-Facebook", the parallel time scaling is far superior on the mesh.
\subsection{Scaling}
We observe in Figure~\ref{fig:scaling} the runtime scaling of our algorithm for the ``ego-Facebook" and ``Stanford bunny" graphs. We can see that, the Facebook graph has drastically worse scaling properties which level out well before reaching the final 2048 thread test. Additionally the Facebook graph requires far more time to coarsen, despite it being almost a tenth the size of the bunny mesh. In comparison, the bunny mesh experiences nearly theoretically ideal strong scaling. This scaling difference can be attributed to load imbalance. Note that in our parallel algorithm we do not parallelize within each norm computation. For a highly irregular network such as ``ego-Facebook" algorithm~\ref{alg:approxgreedy} partitions edges naively across threads without any consideration for the nodes at either end of each edge. The irregularity of the degree distribution in ``ego-Facebook" guarantees that norm-differences along certain edges will require far more compute time than others, and by not accounting for this it is likely that some threads have significant work to perform when computing norms, while others have little to compute. As an example, in the norm computation step one thread may own $k$ edges which connect to the largest degree node in the network with degree $D$, while another may own $k$ edges that all connect to nodes with unit degree. In this case the latter thread only has to complete $2k$ operations, while the former thread has to complete more than $k(D+1)$ operations. This can be mitigated in future work through more sophisticated thread parallelism and usage of more modern GPU architectures, where warp divergence of threads is less of an issue~\cite{coorporation2017nvidia}.
\begin{figure}
    \centering
    \includegraphics[width=0.4\textwidth]{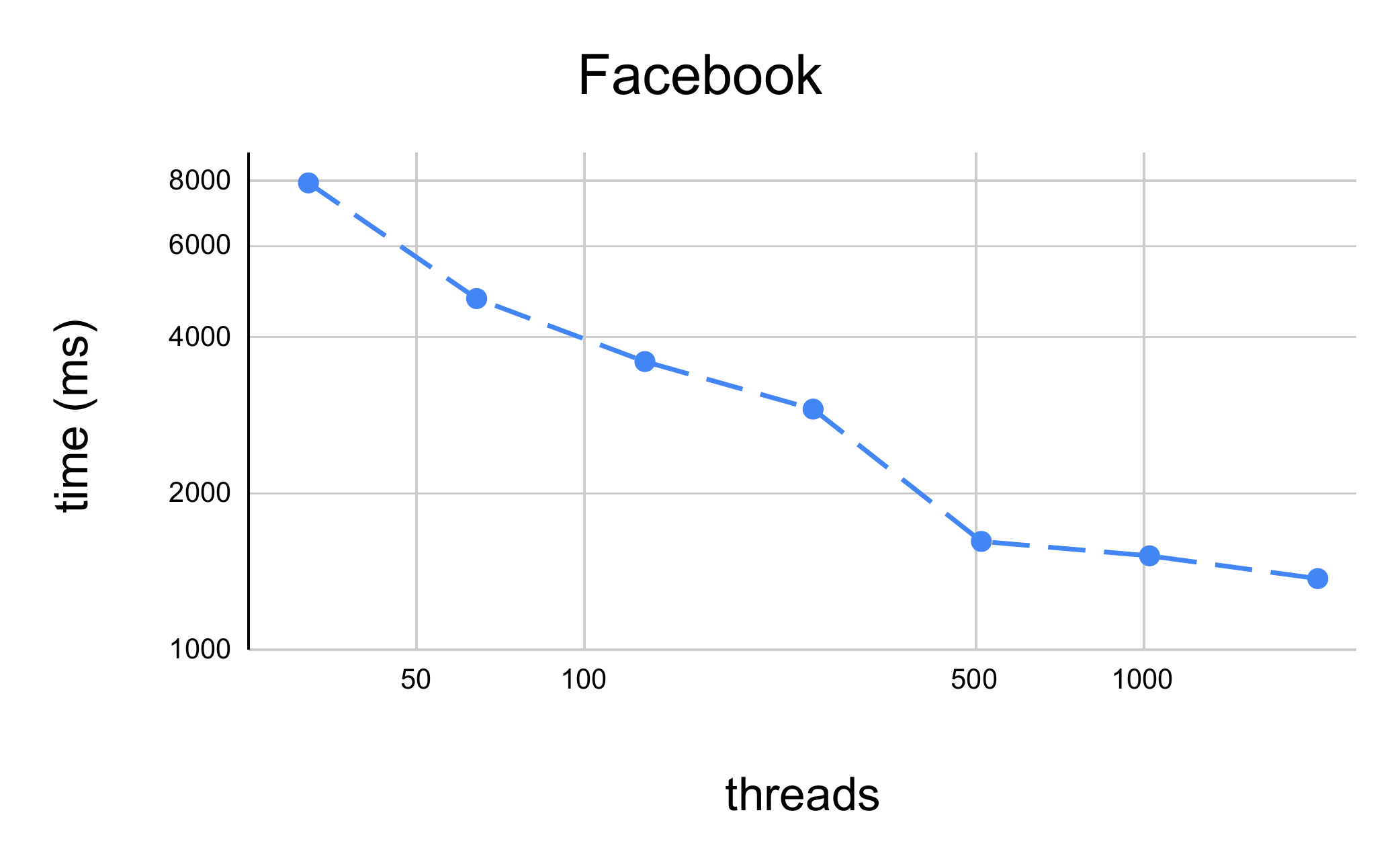}
    \includegraphics[width=0.4\textwidth]{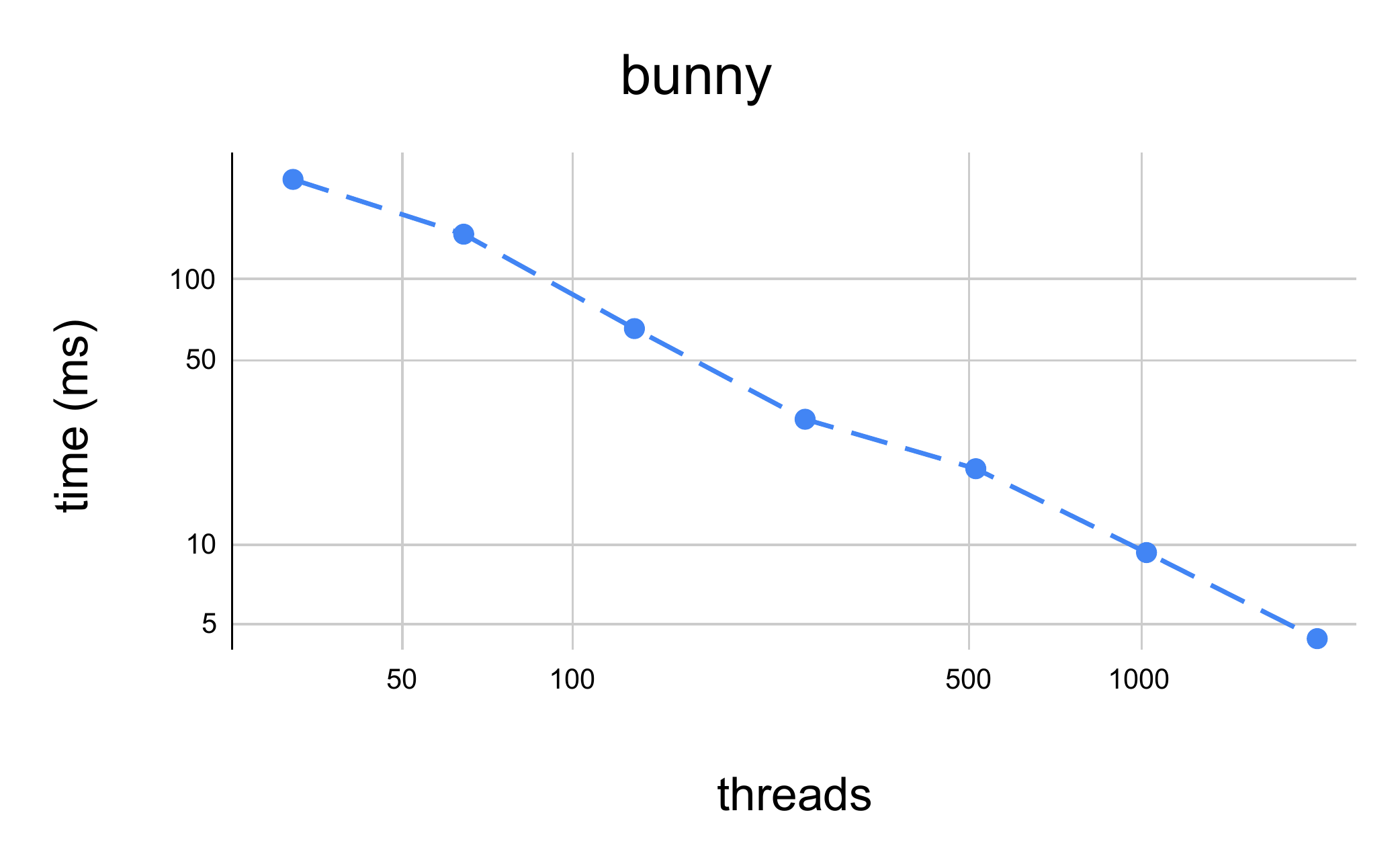}
    \caption{Here we compare the scaling properties of our algorithm running from 32 to 2048 threads on a single Nvidia Quatro M5000 on each of the two test graphs. We can see that the runtimes on the bunny mesh is incredibly small in comparison with the Facebook graph, which has a highly irregular degree distribution. Also notably, we observe that the bunny mesh almost reaches theoretical perfect scaling.}
    \label{fig:scaling}
\end{figure}

\subsection{Spectral approximation}
In Figure~\ref{fig:sectralcomp} we present the eigenvector and eigenvalue approximation properties for the ``Stanford bunny" as well as ``ego-Facebook" on the first fifty nontrivial eigen-pairs. We compare the approximation properties on graphs coarsened to half the number of nodes, quarter the number of nodes, and an eighth the number of nodes. We can see that for both graphs, the half-coarsening is a reasonably accurate approximation of the original eigenvalues, however the eigenvalues quickly stray as we add additional levels of coarsening. As for eigenvector approximations, we consider the inner product between the first fifty normalized eigenvectors of the original graph with the first fifty normalized eigenvectors of the lifted graph for each level of coarsening. In Figure~\ref{fig:sectralcomp} eigenvector comparisons comprise the rightmost three columns. The better the eigenvector approximation is, the more each matrix will resemble an identity matrix since the eigenvectors of the normalized Laplacian are orthogonal~\cite{chung1997spectral}. We see that the eigenvector approximation appears better for ``ego-Facebook" and for the ``Stanford bunny" becomes rather poor after half coarsening. \\
\indent We also show example eigenvectors on the bunny graph and the lift of the half-coarsened bunny graph in Figure~\ref{fig:eigenvectors}. We can see that, while the positive and negative regions drift for the selected eigenvectors, the overall pattern is somewhat conserved. This is of significant interest since the sign patterns of the eigenvectors are important for tasks such as spectral clustering~\cite{zaki2014data}.
\begin{figure}
    \centering
    \includegraphics[width=\textwidth]{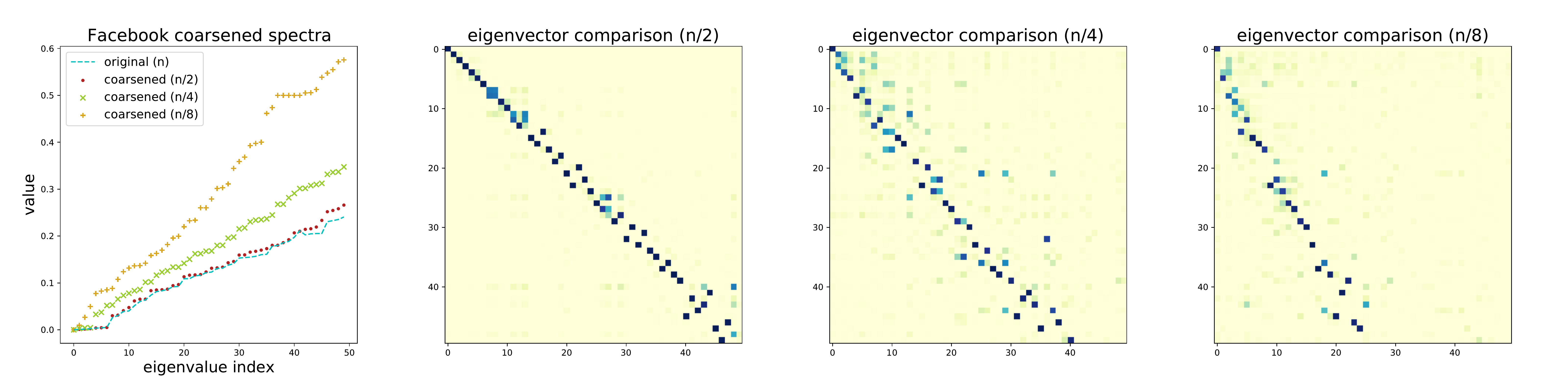}
    \includegraphics[width=\textwidth]{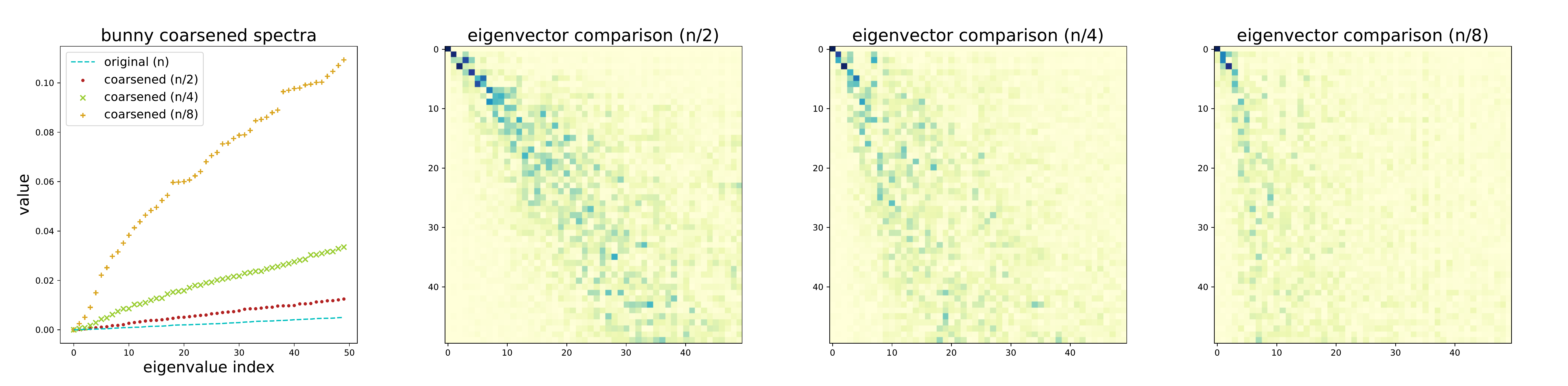}
    \caption{We present spectral approximation properties on the two test-graphs. Each graph was coarsened to half size, quarter size, and eighth size. On the left we compare the first 50 eigenvalues of each lift against the eigenvalues of the original graph. The right three columns compare the angles between the associated eigenvectors by considering the dot product of the eigenvectors in the original graph, with the eigenvectors of each lifted graph. Intuitively, the closer this matrix resembles the identity, the better the eigenvector approximation is. }
    \label{fig:sectralcomp}
\end{figure}

\begin{figure}
    \centering
    \includegraphics[width=0.8\textwidth]{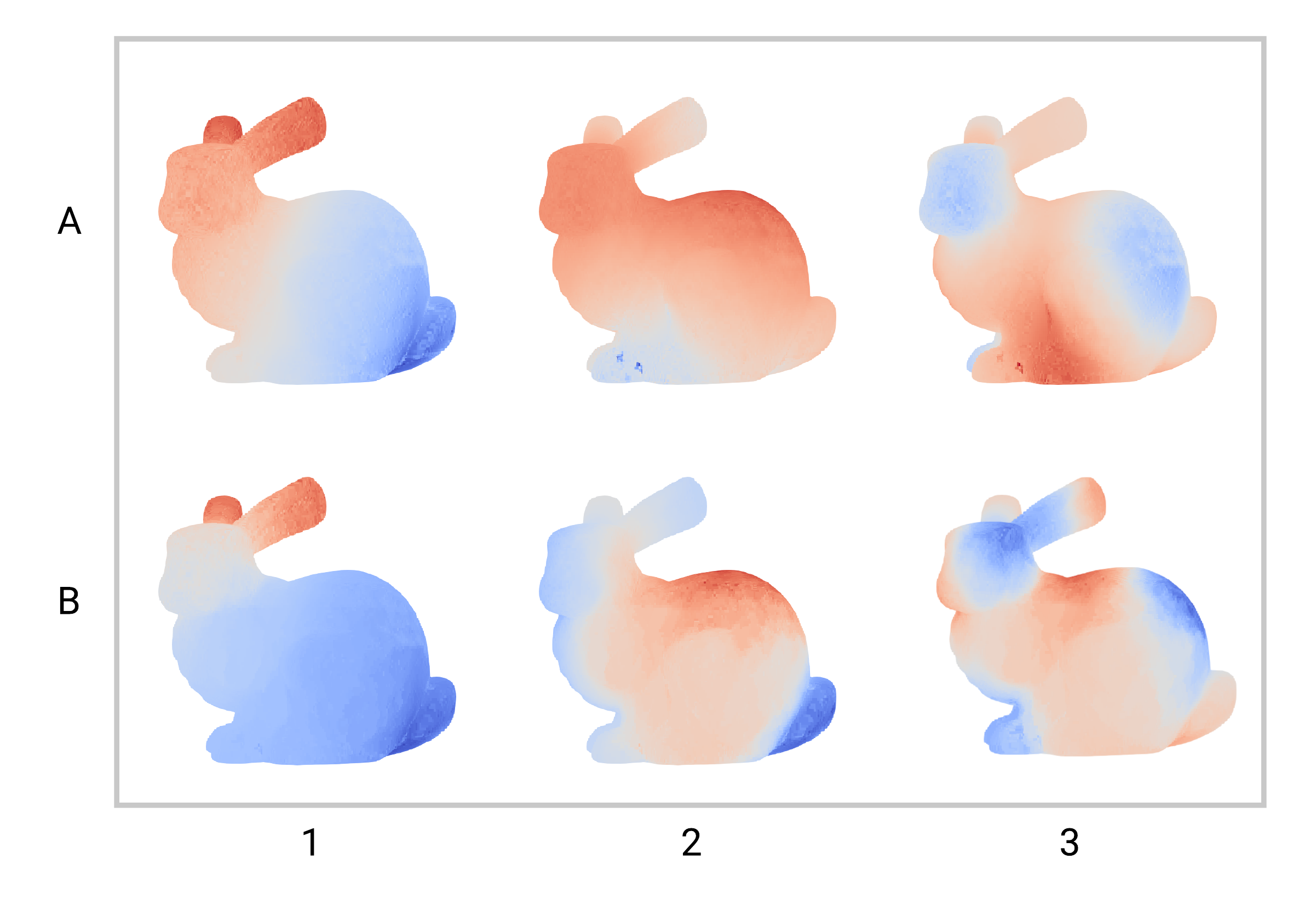}
    \caption{Here we plot three eigenvectors on the original bunny mesh as well as the corresponding lifted eigenvectors on the half-coarsened bunny mesh. The \textbf{A} row are the eigenvectors on the original mesh, and the \textbf{B} row contains the eigenvectors of the coarsened mesh. From left to right the columns correspond to the Fiedler vector, the fifth, and the tenth nontrivial eigenvectors respectively. The blue portions denote negative values and red portions denote positive values. We can see that the eigenvector values have drifted even though the eigenvalues are reasonably close as is seen in Figure~\ref{fig:sectralcomp}.}
    \label{fig:eigenvectors}
\end{figure}

\section{Conclusion and further work}
In this paper we iterated upon a previously researched algorithm for coarsening graphs while preserving approximate eigenvalues. This method can be used to quickly compute spectrally approximate coarsenings of general graph data such as geometric and finite element meshes. We presented a new bound for the eigenvalue differences between original and coarsened graphs, and presented a parallel algorithm for performing graph coarsening within this bound. Our suggested algorithm has a parallel time of $O( \frac{n}{p}\left<d^2\right> + \frac{m}{p}log(m) + (n-n_c) )$ for coarsening graph $G=(V,E)$ to $n_c$ nodes, a vast theoretical speed up over other previously suggested coarsening algorithms preserving spectra. We additionally showed that our algorithm remains faithful to the spectrum of the original graph for limited amounts of coarsening, however the spectrum deviates significantly as the number of nodes in the coarsening diminishes. \\
\indent Our proposed coarsening algorithm leaves room for improvement in the areas of accuracy and parallel time scaling. For the former, we observe that our Bound~\ref{thm:approx2} is rather loose, and it may be tightened by recomputing difference-norms between nodes after some amount of coarsening. This suggests that greater overall accuracy may be obtainable by iterating the difference-norm computations after some amount of coarsening. The optimal amount of coarsening between such iterations remains an open problem. For parallel scaling, several improvements could be made. As mentioned earlier in the manuscript, we partition edges naively among threads which leads to load balancing issues for irregular degree distributions. This could be remedied by considering more sophisticated methods for partitioning the edge set, however this will likely come at the cost of additional preemptive computations. Additionally, faster norm-computations may be achieved by using finer-grained parallelism, where multiple threads are assigned to compute each individual norm. Currently, each norm is computed sequentially by a single thread, which leads to warp divergence and compounds our imbalance issues. Well-known techniques such as hierarchical parallelism~\cite{hong2011accelerating}, loop collapse~\cite{slota_ipdps2015}, and graph adjacency reordering~\cite{kreutzer2014unified} are relatively straightforward to implement and will greatly improve speedups for irregular graph datasets.

\newpage
\bibliography{ref}
\end{document}